\documentclass{article}
\usepackage{verbatim}
\usepackage{amssymb}
\usepackage{natbib}
\usepackage{graphicx}
\newtheorem{theorem}{Theorem}
\newtheorem{definition}{Definition}
\newtheorem{proposition}{Proposition}

\newtheorem{corollary}{Corollary}
\newenvironment{proof}{\paragraph{Proof:}}{\hfill$\square$}
\begin{document}

\title{Computing Equilibria with Partial Commitment\thanks{I dedicate this
    paper to my sister Jessica, her fianc\'e Jeremy, and their upcoming
    full commitment.  I wish them a lifetime of happiness.}}  \author{Vincent Conitzer}
\author{Vincent Conitzer\\Duke University, Durham, NC, USA}
\date{}
\maketitle
\begin{abstract}
  In security games, the solution concept commonly used is that of a
  Stackelberg equilibrium where the defender gets to commit to a mixed
  strategy.  The motivation for this is that the attacker can repeatedly
  observe the defender's actions and learn her distribution over actions,
  before acting himself.  If the actions were not observable, Nash (or
  perhaps correlated) equilibrium would arguably be a more natural solution
  concept.  But what if some, but not all, aspects of the defender's
  actions are observable?  In this paper, we introduce solution concepts
  corresponding to this case, both with and without correlation.  We study
  their basic properties, whether these solutions can be efficiently
  computed, and the impact of additional observability on the utility
  obtained.
\end{abstract}

\section{Introduction}

Algorithms for computing game-theoretic solutions have long been of
interest, but were for a long time not deployed in real-world applications
(at least if we do not count, e.g., computer poker programs---for an
overview of those, see~\citet{Sandholm10:State}---as real-world
applications).  This changed in 2007 with a series of deployed
applications coming out of Milind Tambe's TEAMCORE research group at the
University of Southern California.  The games in question are what are now
called security games, where a defender has to allocate limited resources
to defend certain targets or patrol a certain area, and an attacker chooses
a target to attack.  The deployed applications include airport
protection~\citep{Pita08:Deployed}, assigning Federal Air Marshals to
flights~\citep{Tsai09:IRIS}, patrolling in ports~\citep{An12:PROTECT}, fare
inspection in transit systems~\citep{Yin12:TRUSTSAIMAG}, and patrolling to
prevent wildlife poaching~\citep{Fang16:Deploying}.

While most of the literature on computing game-theoretic solutions has
focused on the computation of Nash equilibria---including the breakthrough
result that even computing a single Nash equilibrium is
PPAD-complete~\citep{Daskalakis09:Complexity,Chen09:Settling}---in the
security games applications the focus is instead on computing an optimal
mixed strategy to commit to~\citep{Conitzer06:Computing}.  In this model,
one player (in security games, the defender) chooses a mixed strategy, and
the other (the attacker) observes this mixed strategy and best-responds to
it.  This sometimes helps, and never hurts, the former
player~\citep{Stengel10:Leadership}.  Intriguingly, in two-player
normal-form games, such a strategy can be computed in polynomial time via
linear
programming~\citep{Conitzer06:Computing,Stengel10:Leadership}. Another
benefit of this model is that it sidesteps issues of equilibrium selection
that the approach of computing (say) a Nash equilibrium might face.

Such technical conveniences aside, the standard motivation for assuming
that the defender in security games can commit to a mixed strategy is as
follows. The defender has to choose a course of action every day. The
attacker, on the other hand, does not, and can observe the defender's
actions over a period of time. Thus, the defender can establish a
reputation for playing any particular mixed strategy. This can be
beneficial for the defender: whereas in a simultaneous-move model (say,
using Nash equilibrium as the solution concept), she can play only best
responses to the attacker's strategy, in the commitment model she can
commit to play something that is not a best response, which may incentivize
the attacker to play something that is better for the defender.  Of course,
for this argument to work, it is crucial that the attacker observes over
time which actions the defender takes before taking any action
himself. Previous work has questioned this and considered models where
there is uncertainty about whether the attacker observes the defender's
actions at all~\citep{Korzhyk11:Stackelberg,Korzhyk11:Solving}, as well as
models where the attacker only gets a limited number of
observations~\citep{Pita10:Robust,An12:Security}.

In this paper, we consider a different setting where some defender actions
are (externally) indistinguishable from each other.  This captures, for
example, the case where there are both observable and unobservable security
measures, as is often the case.  Here, two courses of action are
indistinguishable if and only if they differ only in the unobservable
component.  It also captures the case where a guard can be assigned to a
visible location (1), or to one of two invisible locations (2 or 3).  In
this case, the first action is distinguishable from the latter two, but the
latter two are indistinguishable from each other.  Indistinguishability is
an equivalence relation that partitions the player's strategy space; we
call one element of this partition a SIS (subset of indistinguishable
strategies).  Thus, the defender can establish a reputation for playing a
particular distribution over the SISes.  However, she cannot establish any
reputation for how she plays {\em within} each SIS, because this is not
externally observable. Thus, intuitively, when the defender plays from a
particular SIS, she needs to play a strategy that, within that SIS, is a
best response; however, if there is another strategy in a {\em different}
SIS that is a better response, that is not a problem, because deviating to
that strategy would be observable.

The specific contributions of this paper are as follows.  We formalize
solution concepts for these settings that generalize both Nash and
correlated equilibrium, as well as the basic Stackelberg model with (full)
commitment to mixed strategies.  Further contributions include illustrative
examples of these solutions, basic properties of the concepts, analysis of
their computational complexity, and analysis of how the row player
(defender)'s utility varies as a function of the amount of commitment power
(as measured by observability).

\section{Definitions and Basic Properties}

We are now ready to define some basic concepts.  Throughout, the row player
(player 1) is the player with (some) commitment power, in the sense of
being able to build a reputation. $R$ denotes the set of rows, $C$ the set
of columns, and $\sigma_1$ and $\sigma_2$ denote mixed strategies over
these, respectively.

\begin{definition}
  A {\em subset of indistinguishable strategies (SIS)} $S$ is a maximal
  subset of $R$ such that for any two rows $r_1,r_2 \in S$, the column
  player's observation is identical for $r_1$ and $r_2$.  Let $\mathcal{S}$
  denote the set of all SISes, constituting a partition of $R$.  Given a
  mixed strategy $\sigma_1$ for the row player and a SIS $S$, let
  $\sigma_1(S) = \sum_{r \in S} \sigma_1(r)$ (where $\sigma_1(r)$ is the
  probability $\sigma_1$ puts on $r$).
\end{definition}

Since our focus is on games in which one player can build up a reputation
and the other cannot, we do not consider SISes for the column player.
Equivalently, we consider all the column player's strategies to be in the
same SIS.

\begin{definition}
  Two mixed strategies $\sigma_1, \sigma_1'$ are {\em indistinguishable} to the
  column player if for all $S \in \mathcal{S}$,
  $\sigma_1(S)=\sigma_1'(S)$.
\end{definition}

\noindent {\bf Example.} Consider the following game:

\begin{center}
\begin{tabular}{r@{\hspace{1em}}|@{\hspace{1em}}c@{\hspace{1em}}|@{\hspace{1em}}c@{\hspace{1em}}|}
  &  $A$ &  $B$ \\ \hline
$a$ & 7,0 & 2,1 \\ \hline
$b$ & 6,1 & 0,0 \\ \hline
$c$ & 5,0 & 0,1 \\ \hline
$d$ & 4,1 & 1,0 \\ \hline
\end{tabular}\\
\end{center}
If the players move simultaneously, then $a$ is a strictly dominant
strategy and we obtain $(a,B)$ as the iterated strict dominance solution
(and hence the unique Nash equilibrium), with a utility of $2$ for the row
player.  If the row player gets to commit to a mixed strategy, then she
could commit to play $a$ and $b$ with probability $1/2$ each, inducing the
column player to play $A$,\footnote{As is commonly assumed in this model,
  ties for the column player are broken in the row player's favor; if not,
  the row player can simply commit to $1/2-\epsilon$ on $a$ and
  $1/2+\epsilon$ on $b$.} resulting in a utility of $6.5$ for the row
player. (Even committing to a pure strategy---namely, $b$---would result in
a utility of $6$.) Now suppose $\mathcal{S}=\{\{a,b\},\{c,d\}\}$, i.e., $a$
and $b$ are indistinguishable and so are $c$ and $d$.  In this case,
playing $a$ and $b$ with probability $1/2$ each (or playing $b$ with
probability $1$) is indistinguishable from playing $a$ with probability
$1$.  Hence, it is not credible that the row player would ever play $b$,
given that $a$ is a strictly dominant strategy.  But can the row player
still do better than always playing $a$ (and thereby inducing the column
player to play $B$)?  

We will return to this example shortly, but first we need to formalize the
idea of a deviation that cannot be detected by the column player.

\begin{definition}
  A profile $(\sigma_1,\sigma_2)$ has {\em no undetectable beneficial
    deviations} if (1) for all $\sigma_2'$, $u_2(\sigma_1,\sigma_2') \leq
  u_2(\sigma_1,\sigma_2)$, and (2) for all $\sigma_1'$ indistinguishable
  from $\sigma_1$, $u_1(\sigma_1',\sigma_2) \leq u_1(\sigma_1,\sigma_2)$.
\label{def:undetectable}
\end{definition}

The following simple proposition points out that this is equivalent to the
column player only putting probability on best responses, and the row
player only putting probability on rows that {\em within their SIS} are
best responses.

\begin{proposition}
  A profile $(\sigma_1,\sigma_2)$ has no undetectable beneficial deviations
  if and only if (1) for all $c, c' \in C$ with $\sigma_2(c)>0$,
  $u_2(\sigma_1,c') \leq u_2(\sigma_1,c)$, and (2) for all $S \in
  \mathcal{S}$, for all $r, r' \in S$ with $\sigma_1(r)>0$,
  $u_1(r',\sigma_2) \leq u_1(r,\sigma_2)$.
\end{proposition}

\noindent {\bf Example continued.}  In the game above, consider the
profile $$(((1/2)c,(1/2)d),((1/2)A,(1/2)B))$$ This profile has no
undetectable deviations: (1) the column player is best-responding, and (2)
the only undetectable deviations for the row player do not put any
probability on $\{a,b\}$, and $c$ and $d$ are both equally good responses.

Note that a profile that has no undetectable beneficial deviations may
still not be stable, in the sense that player $1$ may prefer to deviate to
a mixed strategy that is in fact distinguishable from $\sigma_1$, and build
up a reputation for playing that strategy instead. But in a sense, these
profiles are {\em feasible} solutions for the row player: {\em given} that
the row player decides to build up a reputation for the distribution over
SISes resulting from $\sigma_1$, the profile $(\sigma_1,\sigma_2)$ is
stable.  This is similar to the sense in which in the regular Stackelberg
model, any profile consisting of a mixed strategy for the row player and a
best response for the column player is feasible: the row player may not
have had good reason to commit to that particular mixed strategy, but {\em
  given} that she did, the profile is stable.  In fact, this just
corresponds to the special case of our model where all rows are
distinguishable.

\begin{proposition}
  If $|\mathcal{S}|=1$ (all rows are indistinguishable), then a profile has
  no undetectable beneficial deviations if and only if it is a Nash
  equilibrium of the game.  If $|\mathcal{S}|=|R|$ (all rows are
  distinguishable), then a profile has no undetectable beneficial deviations
  if and only if the column player is best-responding.
\label{prop:UBD}
\end{proposition}

We can now define an optimal solution.

\begin{definition}
  A profile with no undetectable beneficial deviations is a {\em
    Stackelberg equilibrium with limited observation (SELO)} if among such
  profiles it maximizes the row player's utility.
\end{definition}

\noindent {\bf Example continued.}  In the game above, consider the
profile $$(((1/2)a,(1/2)d),((1/2)A,(1/2)B))$$ This profile has no
undetectable deviations: $A$ and $B$ are both best responses for the column
player, and the row player strictly prefers $a$ to $b$ and is indifferent
between $c$ and $d$.  It gives the row player utility $3.5$. We now argue
that it is in fact a SELO.  First, note that a SELO must put at least
probability $1/2$ on $d$: for, if it did not, then, because the row player
would never play $b$, the column player would strictly prefer $B$, which
would result in lower utility for the row player.  Second, the column
player must play $B$ at least half the time, because otherwise, the row
player would strictly prefer $c$ to $d$---but if the row player only plays
$a$ and $c$, the column player would strictly prefer $B$.  Under these two
constraints, the row player would be best off having as much as possible of
the remaining probabilities on $a$ and $A$, and this results in the profile
above.
 
\begin{proposition}
  If $|\mathcal{S}|=1$ (all rows are indistinguishable), then a profile is
  a SELO if and only if it is a Nash equilibrium that maximizes the row
  player's utility among Nash equilibria.  If $|\mathcal{S}|=|R|$ (all rows
  are distinguishable), then a profile is a SELO if and only if it is a
  Stackelberg equilibrium (with full observation).
\label{prop:SELO}
\end{proposition}

\section{Computational Results}

We now consider the complexity of computing a SELO.  We immediately obtain:

\begin{corollary}
  When $|\mathcal{S}|=1$, computing a SELO is NP-hard (and the maximum
  utility for the row player in a profile with no undetectable beneficial
  deviations is inapproximable unless P=NP).
\end{corollary}
\begin{proof}
  By Propositions~\ref{prop:UBD} and~\ref{prop:SELO}, these problems are
  equivalent to maximizing the row player's utility in a Nash equilibrium,
  which is known to be NP-hard and
  inapproximable~\citep{Gilboa89:Nash,Conitzer03:Nash}.
\end{proof}

This still leaves open the question of whether the problem becomes easier
if the individual SISes have small size. Unfortunately, the next result
shows that the problem remains NP-hard and inapproximable in this case.
This motivates extending the model to one that allows correlation, as we
will do in Section~\ref{se:signaling}.

\begin{theorem}
  Computing a SELO remains NP-hard even when $|S|=2$ for all $S \in
  \mathcal{S}$ (and in fact it is NP-hard to check whether there exists a
  profile with no undetectable beneficial deviations that gives the row
  player positive utility, even when all payoffs are nonnegative).
\label{th:hard}
\end{theorem}
\begin{proof}
  We reduce from the EXACT-COVER-BY-3-SETS problem, in which we are given a
  set of elements $T$ ($|T|=m$, with $m$ divisible by $3$) and subsets $T_j
  \subseteq T$ that each satisfy $|T_j|=3$, and are asked whether there
  exist $m/3$ of these subsets that together cover all
  of $T$.  For an arbitrary instance of this problem, we construct the
  following game.  For each $T_j$, we add a SIS consisting of two rows,
  $\{T_j^+,T_j^-\}$, as well as a column $T_j$.  For each element $t \in
  T$, we add a column $t$.  The utility functions are as follows.
\begin{itemize}
\item $u_1(T_j^+,T_j)=m/3$ for any $j$
\item $u_1(T_j^+,T_{j'})=0$ for any $j,j'$ with $j \neq j'$
\item $u_1(T_j^-,T_{j'})=1$ for any $j,j'$
\item $u_1(r, t)=0$ for any row $r$ and element $t$
\item $u_2(r,T_j)=m/3-1$ for any row $r$ and any $j$
\item $u_2(T_j^+,t)=0$ for any $j$ and $t \in T_j$
\item $u_2(r,t)=m/3$ for any element $t$ and row $r$ that is not some $T_j^+$ with $t \in T_j$
\end{itemize}
First suppose the EXACT-COVER-BY-3-SETS instance has a solution.  Let the row player play uniformly over the $m/3$ corresponding rows $T_j^+$, and the column player uniformly over the $m/3$ corresponding columns $T_j$.  The row player's expected utility for any of the rows in her support is $1$; deviating to the corresponding $T_j^-$ would still only give her $1$.  The column player's expected utility is $m/3-1$ for any $T_j$; because the row player plays an exact cover, deviating to any $t$ gives him expected utility $(m/3)(m/3-1)/(m/3)=m/3-1$.  
So this profile has no undetectable beneficial deviations (in fact it is a Nash equilibrium)
 and gives the row player an expected utility of $1$.

Now suppose that the game has a SELO in which the row player gets positive utility, which implies that the 
column player puts total probability $p>0$ on his $T_j$ columns.  It follows that for every $t \in T$, the total probability that the row player puts on rows $T_j^+$ with $t \in T_j$ is at least $3/m$, or otherwise the column player would strictly prefer playing $t$ to playing any $T_j$. However, note that the row player can only put positive probability on rows $T_j^+$ where the corresponding column $T_j$ receives probability at least $3p/m$ (thereby resulting in expected utility at least $p$ for the row player for playing $T_j^+$), because otherwise the corresponding row $T_j^-$ (which is indistinguishable) would be strictly preferable (resulting in expected utility $p$).  But of course there can be at most $m/3$ such columns $T_j$, and these $T_j$ must cover all the elements $t$ by what we said before.  Hence the EXACT-COVER-BY-3-SETS instance has a solution.
\end{proof}

\section{Adding Signaling}
\label{se:signaling}

The notion of correlated equilibrium~\citep{Aumann74:Subjectivity} results
from augmenting a game with a trusted mediator that sends correlated
signals to the agents.  As is well known, without loss of generality, we
can assume the signal that an agent receives is simply the action she is to
take.  This is for the following reason.  If a correlated equilibrium
relies on an agent randomizing among multiple actions conditional on
receiving a particular signal, then we may as well have the mediator do
this randomization on behalf of the agent before sending out the signal. It
is well known that correlated equilibria can outperform Nash equilibria
from all agents' perspectives.  For example, consider Shapley's game, which
is a version of rock-paper-scissors where choosing the same action as the
other counts as a loss.

\begin{center}
\begin{tabular}{r@{\hspace{1em}}|@{\hspace{1em}}c@{\hspace{1em}}|@{\hspace{1em}}c@{\hspace{1em}}|@{\hspace{1em}}c@{\hspace{1em}}|}
  &   $A$ & $B$ & $C$ \\ \hline
$a$ & 0,0 & 1,0 & 0,1 \\ \hline
$b$ & 0,1 & 0,0 & 1,0 \\ \hline
$c$ & 1,0 & 0,1 & 0,0 \\ \hline
\end{tabular}\\
\end{center}

Whereas the only Nash equilibrium of this game is for both players to
randomize uniformly (resulting in $0,0$ payoffs $1/3$ of the time), there
is a correlated equilibrium that only results in the $1,0$ and $0,1$
outcomes, each $1/6$ of the time.  That is, if the mediator is set up to
draw one of these six entries uniformly at random, and then tell each agent
what she is supposed to play (but not what the other is supposed to play),
then each agent has an incentive to follow the recommendation: doing so
will result in a win half the time, and it is not possible to do better
given what the agent knows.

Correlated equilibria are easier to compute than Nash equilibria: given a
game in normal form, there is a linear program formulation for computing
even optimal correlated equilibria (say, ones that maximize the row
player's utility).  The linear program presented later in
Figure~\ref{fi:LP} is closely related.

Similar signaling has received attention in the Stackelberg model.  One may
assume a more powerful leader in this model that can commit not only to
taking actions in a particular way, but also to sending signals in a way
that is correlated with how she takes actions.  (Again, the motivation for
using this in real applications might be that over time the leader develops
a reputation for sending out signals according to a particular
distribution, and playing particular distributions over actions conditional
on those signals.)  Because the leader can commit to sending signals in a
particular way, there is no need to introduce an independent mediator
entity in this context.  As it turns out, in a two-player normal-form game
this additional power does not buy the leader anything, but with more
players it does~\citep{Conitzer11:Commitment}. Such signaling can also help
in Bayesian games~\citep{Xu16:Signaling} and stochastic
games~\citep{Letchford12:Computing}, both from the perspective of
increasing the leader's utility and from the perspective of making the
computation easier.

It is straightforward to see that signaling can be useful in our limited
commitment model as well.  For example, if we just take Shapley's game with
$|\mathcal{S}|=1$, then by Proposition~\ref{prop:SELO} without signaling we
are stuck with the Nash equilibrium, but it seems we should be able to
obtain the improved correlated equilibrium outcome with some form of
signaling.  But what is the right model of signaling here?  
We start with a very powerful model of signaling, and then discuss less
powerful models in Section~\ref{se:weaker}.

\begin{definition}
  In the {\em trusted mediator model}, the row player can design an
  independent trusted mediator that sends signals privately to each player
  according to a pre-specified joint distribution.  After the round of play
  has completed, the mediator publicly reveals the signal sent to the row
  player.
\label{def:trusted}
\end{definition}

The after-the-fact public revelation of the signal sent to the row player
allows the row player to commit to (i.e., in the long run develop a
reputation for) responding to each signal with a particular distribution of
play.  Specifically, after each completed round, the column player learns
the signal sent to, and the SIS played by, the row player.\footnote{It is
  easy to get confused here---does the column player not learn more in a
  round purely by virtue of his own payoff from that round?  It is
  important to remember that we are not considering repeated play by the
  column player. The idea is that the column player can observe over time
  the signals and how the row player acts {\em before} the column player
  ever acts.  For discussion of security contexts in which certain types of players
  can receive messages that are inaccessible to other types,
  see~\citet{Xu16:Signaling}.}  Thus, if the row player according to the
signal that she received was supposed to play an action from a particular
SIS, then the column player can verify that she did.  However, the row
player may have an incentive to deviate {\em within} a SIS, because this is
undetectable.

In the appendix, we show that under the trusted mediator model, without
loss of generality a signal consists of just an action to play.
With this in mind, we now define formally what it means for a correlated
profile to have no undetectable beneficial deviations.

\begin{definition}
A correlated profile $\sigma$ has {\em no undetectable beneficial
  deviations} if (1) for all $c,c' \in C$ with $\sum_{r \in R}
\sigma(r,c)>0$, we have $\sum_{r \in R} \sigma(r,c) (u_2(r,c)-u_2(r,c'))
\geq 0$, and (2) for all $S \in \mathcal{S}$, for all $r,r' \in S$ with
$\sum_{c \in C} \sigma(r,c)>0$, we have $\sum_{c \in C}
\sigma(r,c)(u_1(r,c)-u_1(r',c)) \geq 0$.
\label{def:undetectable2}
\end{definition}

Note that, as is well known in the formulation of correlated equilibrium,
in the first inequality, we can use $\sigma(r,c)$ rather than the more
cumbersome $\sigma(r,c) / \sum_{r'' \in R} \sigma(r'',c)$, which would be
the conditional probability of seeing $r$ given a signal of $c$, because
the denominator is a constant (similar for the second inequality).  As a
result, the condition that $\sum_{r \in R} \sigma(r,c)>0$ is in fact not
necessary because the inequality is vacuously true otherwise.  This is what
allows the standard linear program formulation of correlated equilibrium,
as well as the linear program we present below in Figure~\ref{fi:LP}.

\begin{definition}
A correlated profile with no undetectable beneficial deviations is a {\em
  Stackelberg equilibrium with signaling and limited observation (SESLO)}
if among such profiles it maximizes the row player's utility.
\end{definition}

\noindent {\bf Example.} Consider the following game:
\begin{center}
\begin{tabular}{r@{\hspace{1em}}|@{\hspace{1em}}c@{\hspace{1em}}|@{\hspace{1em}}c@{\hspace{1em}}|@{\hspace{1em}}c@{\hspace{1em}}|@{\hspace{1em}}c@{\hspace{1em}}|}
    &  $A$ & $B$ & $C$ & $D$\\ \hline
$a$ & 0,0 & 12,0 & 0,1 & 0,0\\ \hline
$b$ & 0,1 & 0,0 & 12,0 & 0,0\\ \hline
$c$ & 12,0 & 0,1 & 0,0 & 0,0\\ \hline
$d$ & 5,0 & 5,0 & 5,0 & 0,1\\ \hline
$e$ & 7,0 & 7,0 & 7,0 & 1,1\\ \hline
\end{tabular}\\
\end{center}
Suppose $\mathcal{S}=\{\{a,b,c,d\},\{e\}\}$.  Then the following correlated
profile (in which the signal an agent receives is which action to take) is
a SESLO:
$$((1/9)(a,B),(1/9)(a,C),(1/9)(b,A),(1/9)(b,C),(1/9)(c,A),(1/9)(c,B),$$ $$(1/9)(e,A),(1/9)(e,B),(1/9)(e,C))$$
With this profile, for any signal the column player can receive, following
the signal will give him utility $1/3$, and so will any deviation.  For any
signal the row player receives in SIS $\{a,b,c,d\}$, following the signal
will give her $6$; deviating to $a$, $b$, or $c$ will give either $0$ or
$6$, and deviating to $d$ will give $5$. The row player obtains utility
$19/3$ from this profile.\footnote{This was verified to be optimal using
  the linear program in Figure~\ref{fi:LP}; same for the next case.}  In
contrast, without any commitment (if $|\mathcal{S}|$ had been $1$), the
outcome $(e,D)$ would have been a SESLO, giving the row player utility only
$1$.  Also, without signaling (but still with
$\mathcal{S}=\{\{a,b,c,d\},\{e\}\}$), the outcome $(e,D)$ would have been a
SELO.  For consider a mixed-strategy profile without any undetectable
beneficial deviations, and suppose it puts positive probability on at least
one of $A$, $B$, and $C$.  Then at least one of $a$, $b$, and $c$ must get
positive probability as well, for otherwise the column player would be
better off playing $D$.  Because $a$, $b$, and $c$ are all in the same SIS
and perform equally well against $D$, and because $A$, $B$, and $C$ all
perform equally well against $d$ and $e$, if we condition on the players
playing from $a$,$b$,$c$ and $A$,$B$,$C$, the result must be a Nash
equilibrium of that $3 \times 3$ game, which means that all of $A$, $B$,
and $C$ get the same probability.  But in that case, $d$ (which is in the
same SIS) is a better response for the row player, and we have a
contradiction.  Hence any SELO involves the column player always playing
$D$ and the most the row player can obtain is $1$.

We next have the following simple proposition that the ability to signal never
hurts the row player.

\begin{proposition}
The row player's utility from a SESLO is always at least that of a SELO.
\end{proposition}
\begin{proof}
  We show that an uncorrelated profile $(\sigma_1,\sigma_2)$ that has no
  undetectable deviations (in the sense of
  Definition~\ref{def:undetectable}) also has no undetectable deviations
  (in the sense of Definition~\ref{def:undetectable2}) when interpreted as
  a correlated profile $\sigma$ (with
  $\sigma(r,c)=\sigma_1(r)\sigma_2(c)$); the result follows.  First, for
  all $c,c' \in C$ with $\sum_{r \in R} \sigma(r,c)>0$ (which is
  equivalent to $\sigma_2(c)>0$), we have $\sum_{r \in R} \sigma(r,c)
  (u_2(r,c)-u_2(r,c')) = \sigma_2(c) \sum_{r \in R} \sigma_1(r)
  (u_2(r,c)-u_2(r,c')) = \sigma_2(c) (u_2(\sigma_1,c) - u_2(\sigma_1,c'))
  \geq 0$ by the best-response condition of
  Definition~\ref{def:undetectable}.  Similarly, for all $S \in
  \mathcal{S}$, for all $r,r' \in S$ with $\sum_{c \in C} \sigma(r,c)>0$
  (which is equivalent to $\sigma_1(r)>0$), we have $\sum_{c \in C}
  \sigma(r,c)(u_1(r,c)-u_1(r',c)) = \sigma_1(r) \sum_{c \in C}
  \sigma_2(c)(u_1(r,c)-u_1(r',c)) = \sigma_1(r)
  (u_1(r,\sigma_2)-u_1(r',\sigma_2)) \geq 0$ by the
  best-response-within-a-SIS condition of
  Definition~\ref{def:undetectable}.
\end{proof}

\begin{proposition}
  If $|\mathcal{S}|=1$ (all rows are indistinguishable), then a profile is
  a SESLO if and only if it is a correlated equilibrium that maximizes the
  row player's utility.  If $|\mathcal{S}|=|R|$ (all rows are
  distinguishable), then a profile is a SESLO if and only if it is a
  Stackelberg equilibrium with signaling (which can do no better than a
  Stackelberg equilibrium without signaling).
\label{prop:SESLO}
\end{proposition}

\section{Computational Results}

It turns out that with signaling, we do not face hardness.  The linear
program in Figure~\ref{fi:LP} can be used to compute a SESLO.  It is a
simple modification of the standard linear program for correlated
equilibrium, the differences being that (1) for the row player, only
deviations within a SIS are considered, and (2) there is an objective of
maximizing the row player's utility.  The special case where
$|\mathcal{S}|=|R|$ has no constraints for the row player, and that special
case of the linear program has previously been described
by~\citet{Conitzer11:Commitment}.

\begin{figure}
\begin{center}
\begin{tabular}{|r l|}
\hline
{\bf maximize} \ & $\sum_{r \in R, c \in C} u_1(r,c)p(r,c)$\\
($\forall \ S \in \mathcal{S}$) ($\forall \ r,r' \in S$) \  &
$\sum_{c \in C} (u_1(r',c)-u_1(r,c))p(r,c) \leq 0$\\
($\forall \ c,c' \in C$) \  &
$\sum_{r \in R} (u_2(r,c')-u_2(r,c))p(r,c) \leq 0$\\
 & $\sum_{r \in R, c \in C} p(r,c)=1$\\
$(\forall \ r \in R, c \in C)$ \ & $p(r,c) \geq 0$\\ \hline
\end{tabular}
\end{center}
\caption{Linear program for computing a SESLO.}
\label{fi:LP}
\end{figure}

\begin{theorem}
A SESLO can be computed in polynomial time.
\label{th:easy}
\end{theorem}

\section{Weaker Signaling Models}
\label{se:weaker}

Another model of signaling would be to have the row player absorb the role
of the mediator as well.  That is, instead of adding an independent entity,
the row player would simply generate and send the signal to the column
player.  This is how the signaling model with full
commitment~\citep{Conitzer11:Commitment} is typically presented, and in the
context of full commitment it makes no difference whether it is the row
player or another independent entity that sends the signals.  However, with
partial commitment, there is a significant difference.  Namely, if the row
player knows which signal the column player will receive, she could base
her own choice of action on this.  Now, if she does so in a way that
changes the distribution over SISes conditional on that column player
signal, this would over time be detected.  Hence, if every SIS consists of
a single row---the full commitment case---then the row player cannot take
advantage of knowing the signal to the column player without being
detected.  However, she can base her choice {\em within} a SIS on the
signal to the column player without ever being detected.  In the most
recent example above, this plays out as follows if we investigate whether
the same correlated profile is still an equilibrium.  By knowing the signal
($A$, $B$, or $C$) to the column player, the row player could, at first
glance, deviate and always obtain utility $12$.  However, if she does not
play $e$ one third of the time, this would be detectable.  Still, if she
continues to play $e$ one third of the time conditional on each of the
signals $A$, $B$, and $C$, but with the remaining $2/3$ plays whichever
action gives her $12$, this would be undetectable and give her an expected
utility of $(1/3) \cdot 7 + (2/3) \cdot 12 = 31/3$, greater than the $19/3$
for not deviating.  Hence this is no longer an equilibrium---the column
player will realize that it is not in his best interest to follow the
signals---and the row player cannot get as high of a utility anymore in
equilibrium.  That is, she benefits from being able to keep herself from
knowing the signal to the column player.  One may wonder whether restricted
signaling from the row player to the column player (where the signal is
perhaps not even a specific column) can still be helpful.  But in fact, as
the next proposition shows, it cannot.  (This was previously established in
the full commitment case~\citep{Conitzer11:Commitment}.)

\begin{proposition}
If the row player necessarily knows the signal sent to the column player,
then the row player can do no better than playing a SELO.
\end{proposition}
\begin{proof}
The signal sent to the column player is common knowledge between the
players.  Hence, in equilibrium, conditional on each signal, they must play
an uncorrelated profile with no undetectable beneficial deviations.  But then the
(correlated) profile as a whole is a convex combination of such
uncorrelated profiles, and hence the row player's utility for it cannot
exceed her utility for the best one among them.
\end{proof}

Hence, to gain any benefit from signaling at all, the row player must
restrict herself from knowing the signal to the column player.

Next, we investigate a more subtle change to the signaling model, namely
the variant where we still have a trusted mediator but the signal to the
row player is not publicly revealed afterwards.  It is perhaps not
immediately clear that this variant is significantly different, because the
column player can still observe the distribution over the SISes conditional
on each signal that he gets.  Hence, it may be the case that the row player
cannot usefully deviate from one SIS to another without this being
detectable, leaving the row player only with the option of deviating within
a SIS, which was already present in the original trusted mediator model.
However, this is not true: the following example shows that deviations
across SISes can be significant when the signal to the row player is not
publicly revealed afterwards.

\noindent {\bf Example.} Consider the following game:
\begin{center}
\begin{tabular}{r@{\hspace{1em}}|@{\hspace{1em}}c@{\hspace{1em}}|@{\hspace{1em}}c@{\hspace{1em}}|@{\hspace{1em}}c@{\hspace{1em}}|@{\hspace{1em}}c@{\hspace{1em}}|}
    & $A$ & $B$ & $C$ & $D$\\ \hline
$a$ & 1,0 & 4,0 & 1,1 & 1,1\\ \hline
$b$ & 1,1 & 1,0 & 4,0 & 1,1 \\ \hline
$c$ & 4,0 & 1,1 & 1,0 & 1,1\\ \hline
$d$ & 1,1 & 0,1 & 0,0 & 0,0 \\ \hline
$e$ & 0,0 & 1,1 & 0,1 & 0,0 \\ \hline
$f$ & 0,1 & 0,0 & 1,1 & 0,0 \\ \hline
\end{tabular}\\
\end{center}
Suppose $\mathcal{S}=\{\{a,b,c\},\{d,e,f\}\}$.
It can be shown\footnote{Again, this was verified using the linear program
  from Figure~\ref{fi:LP}.} that a SESLO is
$$((1/9)(a,B),(1/9)(a,C),(1/9)(b,A),(1/9)(b,C),(1/9)(c,A),(1/9)(c,B),$$ $$(1/9)(d,A),(1/9)(e,B),(1/9)(f,C))$$ resulting in an expected utility of $2$ for the row player. 
(A rough intuition is that the upper left corner of the game is Shapley's
game, and the row player would like to be able to play the profile
corresponding to the correlated equilibrium in Shapley's game.  However,
playing this by itself is not feasible because the column player would then
deviate to $D$.  But this is fixed by mixing in some of the lower left
corner.)  However, now consider the following deviation for the row player.
When recommended to play $d$, $e$, or $f$, instead play $c$, $a$, or $b$,
respectively, resulting in a payoff of $4$ instead of $1$; this happens
$1/3$ of the time overall.  When recommended to play $a$, $b$, or $c$,
instead do the following: with probability $1/2$ play as suggested but with
probability $1/2$ play $e$, $f$, or $d$, respectively, resulting in an
expected payoff of $0.5$ instead of $2.5$; this (being recommended $a$,
$b$, or $c$, and with probability $1/2$ playing differently) happens $1/3$
of the time.  So overall, the expected gain from this deviation is $(1/3)
\cdot 3 + (1/3) \cdot (-2) = 1/3$.  Moreover, conditional on the
recommendation being any one of $A$, $B$, and $C$, the column player still
observes the SIS $\{a,b,c\}$ two thirds of the time and the SIS $\{d,e,f\}$
one third of the time.  Hence the deviation is undetectable, unless it is
revealed after the fact what the row player's signal was.

\section{The Value of More Commitment Power}

More strategies being distinguishable corresponds to more commitment power
for the row player.  As commitment power (in this particular sense)
increases, does the utility the row player can obtain always increase
gradually?  (Note that it can never {\em decrease} the row player's
utility, because all it will do is remove constraints in the optimization.)
If she has close to full commitment power, does this guarantee her most of
the benefit of full commitment power?  Is some nontrivial minimal amount of
commitment power necessary to obtain much benefit from it?  The next two
results demonstrate that the answer to all these questions is ``no'': there
can be big jumps in the utility that the row player can obtain, both on the
side close to full commitment power (Proposition~\ref{prop:closetofull})
and on the side close to no commitment power
(Proposition~\ref{prop:closetonone}).  (For an earlier study comparing the
value of being able to commit completely to that of not being able to
commit at all, see~\citet{Letchford10:Value}; for one assessing the value
of correlation without commitment, see~\citet{Ashlagi08:Value}.)

\begin{proposition}
For any $\epsilon>0$ and any $n>1$, there exists an $n \times (n+1)$ game
with all payoffs in $[0,1]$ such that if $|\mathcal{S}|=|R|=n$, the row
player can obtain utility $1-\epsilon$ (even without signaling), but for any $\mathcal{S}$ with
$|\mathcal{S}|<|R|=n$, the row player can obtain utility at most $\epsilon$
(even with signaling).
\label{prop:closetofull}
\end{proposition}
\begin{proof}
Let $R = \{1,\ldots,n\}$ and $C=\{1,\ldots,n+1\}$.
Let $u_1(i,j)=i\epsilon/n$ for $j \leq n$, and let $u_1(i,n+1) = 1-(n-i)\epsilon/n$.
Let $u_2(i,j)=(1+1/n)/2$ for $i \neq j$ and $j\leq n$, let $u_2(i,i)=0$ (for $i \leq n$), and let $u_2(i,n+1)=1/2$ for all $i$.

Suppose $|\mathcal{S}|=|R|=n$.  Then, by Proposition~\ref{prop:SELO}, we are in the regular Stackelberg model, and the row player can commit to a uniform strategy, putting probability $1/n$ on each $i$.
As a result the expected utility for the column player for playing some $j \leq n$ is $((n-1)/n)(1+1/n)/2 = (n-1)(n+1)/(2n^2) = (n^2-1)/(2n^2) < 1/2$, so to best-respond he needs to play $n+1$, resulting in a utility for the row player that is greater than $1-(n-1)\epsilon/n > 1-\epsilon$.

On the other hand, suppose that $|\mathcal{S}|<|R|=n$.  Hence there exists
some $S \in \mathcal{S}$ with $i,i'\in S$, $i < i'$. Note that $i'$
strictly dominates $i$, so the row player will never play $i$ in a SELO or
even a SESLO.  But then, the column player can obtain $(1+1/n)/2 > 1/2$ by
playing $i$, and hence will not play $n+1$.  As a result the row player
obtains at most $n\epsilon / n = \epsilon$.
\end{proof}

\begin{proposition}
For any $\epsilon>0$ and any $n>1$, there exists an $n \times (n+1)$ game
with all payoffs in $[0,1]$ such that for any $\mathcal{S}$ with
$|\mathcal{S}|>1$, the row player can obtain utility $1-\epsilon$ (even without
signaling), but if $|\mathcal{S}|=1$, the row player can only obtain
utility $0$ (even with signaling).
\label{prop:closetonone}
\end{proposition}
\begin{proof}
Let $R = \{1,\ldots,n\}$ and $C=\{1,\ldots,n+1\}$.
Let $u_1(i,j)=1-\epsilon$ for $i \neq j$ and $j \leq n$, let $u_1(i,i)=1$ (for $i \leq n$),
and let $u_1(i,n+1) = 0$ for all $i$.
Let $u_2(i,j)=1$ for $i \neq j$ and $j\leq n$, let $u_2(i,i)=0$ (for $i \leq n$), and let $u_2(i,n+1)=(n-1/2)/n$ for all $i$.

Suppose $|\mathcal{S}|>1$.  Then, the row player can commit to put $0$
probability on some $S \in \mathcal{S}$, and therefore, $0$ probability on
some $i$.  Hence, this $i$ is a best response for the column player, and
the row player obtains $1-\epsilon$.  (The row player may be able to do
better yet, but this is a feasible solution.)

On the other hand, suppose $|\mathcal{S}|=1$. By
Proposition~\ref{prop:SELO}, the row player can only obtain the utility of
the best Nash equilibrium of the game for her (or, in the case with
signaling, the utility of the best correlated equilibrium, by
Proposition~\ref{prop:SESLO}).  We now show that in every Nash equilibrium
(or even correlated equilibrium) of the game, the column player puts all
his probability on $n+1$, from which the result follows immediately.  For
suppose the column player sometimes plays some $j \leq n$.  Then, for the
row player to best-respond, she has to maximize the probability of choosing
the same $j$ (conditional on the column player playing some $j \leq n$).
(Or, more precisely in the case of correlated equilibrium, conditional on
receiving any signal that leaves open the possibility that the column
player plays some $j \leq n$, the row player has to maximize the
probability of picking the same $j$.)  She can always make this probability
at least $1/n$ by choosing uniformly at random.  Hence, the column player's
expected utility (conditional on playing $j \leq n$) is at most $(n-1)/n$.
But then $n+1$ is a strictly better response, so we do not have a Nash (or
correlated) equilibrium.
\end{proof}

Of course, the above two results are extreme cases.  Can we say anything about
what happens ``typically''?  To illustrate this, we present the results for
randomly generated games in Figure~\ref{fi:experiment}.  For each data
point, 1000 games of size $m \times n$  were generated by choosing utilities
uniformly at random.  The rows were then evenly (round-robin) spread over a
given number of SISes, and the game was solved using the GNU Linear
Programming Kit (GLPK) with the linear program from Figure~\ref{fi:LP}.
The leftmost points ($1$ SIS) correspond to no commitment power (best
correlated equilibrium), and the rightmost points (at least when the number
of SISes is at least $m$)
correspond to full commitment power (best Stackelberg mixed strategy).
\begin{figure}[ht!]
\centering
\includegraphics[]{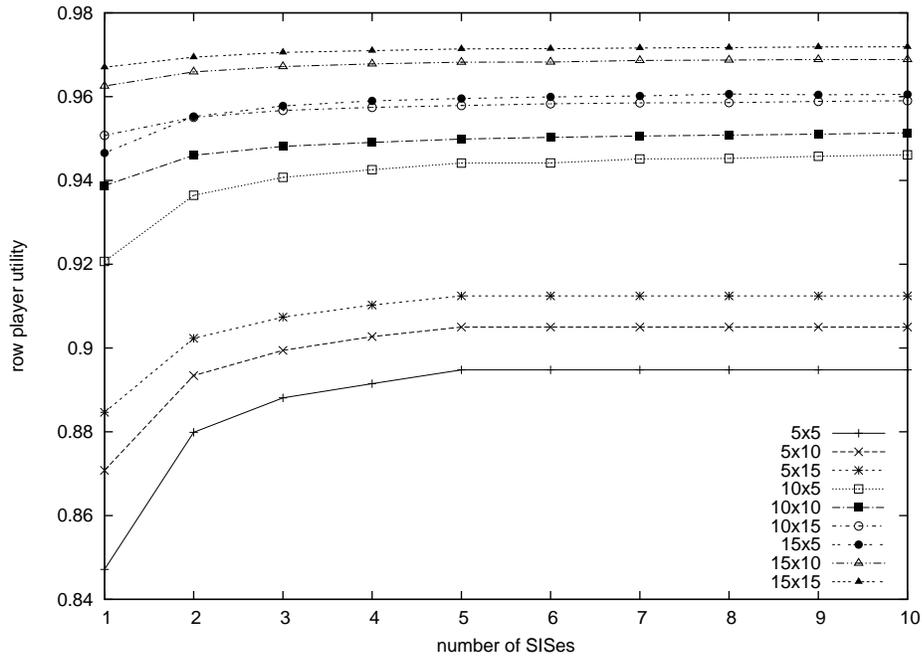}
\caption{Utility obtained by the row player as a function of commitment power
  (number of SISes), for various sizes of $m \times n$ games.}
\label{fi:experiment}
\end{figure}
From this experiment, we can observe that most of the value of commitment is
already obtained when moving from one SIS to two.

\section{Conclusion}

The model of the defender being able to commit to a mixed strategy has been
popular in security games, motivated by the idea that the attacker can
learn the distribution over time.  This model has previously been
questioned, and limited observability has previously been studied in
various senses, including the attacker obtaining only a limited number of
observations~\citep{Pita10:Robust,An12:Security} as well as the attacker
observing (perfectly) only with some
probability~\citep{Korzhyk11:Stackelberg,Korzhyk11:Solving}.  Here, we
considered a different type of limited observability, where certain courses
of action are distinguishable from each other, but others are not.  As a
result, the row player's pure strategies partition into SISes, and she can
commit to a distribution over SISes but not to how she plays within each
SIS.  We showed that it is NP-hard to compute a Stackelberg equilibrium
with limited observation in this context, even when the SISes are small
(Theorem~\ref{th:hard}).  We then introduced a modified model with
signaling and showed that in it, Stackelberg equilibria can be computed in
polynomial time (Theorem~\ref{th:easy}).  We also showed that weaker
signaling models are technically problematic or without power
(Section~\ref{se:weaker}).  Finally, we showed that the cost of introducing
a bit of additional unobservability can be large both when close to full
observability (Proposition~\ref{prop:closetofull}) and close to no
observability (Proposition~\ref{prop:closetonone}); however, in
simulations, introducing a little bit of observability already gives most
of the value of full observability.

Future research may be devoted to the following questions.  Are there
algorithms for computing a SELO that are efficient for special cases of the
problem or that run fast on ``typical'' games?  Another direction for
future work concerns learning in games, which is a topic that has been
thoroughly studied in the simultaneous-move case (see, e.g.,~\citet{Fudenberg98:Theory}), but also already
to some extent in the mixed-strategy commitment case~\citep{Letchford09:Learning,Balcan15:Commitment}.  A model of
learning in games with partial commitment needs to generalize models for
both of these cases.  Finally, can we mathematically prove what is
suggested by the experiment in Figure~\ref{fi:experiment}, namely that in
random games most of the value of commitment is already obtained with only
two SISes?

\section*{Acknowledgments}
I am thankful for support from ARO under grants
W911NF-12-1-0550 and W911NF-11-1-0332, NSF under
awards IIS-1527434, IIS-0953756, CCF-1101659, and CCF-1337215, and a
Guggenheim Fellowship.

\bibliographystyle{named}
\bibliography{/usr/project/conitzer/vccollab/references.bib}

\appendix
\section{Signals Are Actions without Loss of Generality}

In this appendix, we prove that it is without loss of generality to assume
that signals are actions under the trusted mediator model (Definition~\ref{def:trusted}).

\begin{proposition}
Under the trusted mediator model, without loss of generality each signal to
a player is an action that that player is incentivized to follow.
\end{proposition}
\begin{proof}
  The proof strategy is the same as in the usual correlated equilibrium
  case: if for some mediator signal, a player is supposed to randomize over
  actions, then the mediator can simply do this randomization on behalf of
  the player.  Specifically, suppose more generally that there is a set of
  signals $M_i$ for each player $i$, and conditional on receiving $m_i$
  player $i$ is, in equilibrium, supposed to play a distribution
  $\sigma_i(\cdot|m_i)$ over her actions.  Then, instead of sending $m_i$,
  the mediator may as well draw from $\sigma_i(\cdot|m_i)$ and send the
  resulting action $a_ i$.  Such a direct signal in general conveys less
  information to agent $i$ about what signals the other agents received,
  because the randomization from $\sigma_i(\cdot|m_i)$ is not correlated with
  the other signals and so provides no information, and in fact multiple original
  signals $m_i$ can be consistent with the new signal $a_i$, possibly
  resulting in a strict loss of information.  Nevertheless,
  even with this reduced information, player $i$ can obtain the same
  expected utility as before by simply taking the recommended action.
  There is a new twist in the partial commitment model, which is that we
  also need to show that the new signaling mechanism does not increase the
  set of undetectable deviations for the row player.  But it does not:
  again, the row player has less (or the same) information about what the
  others are playing, and moreover now the only undetectable deviation for
  the row player is to always play within the same SIS as the recommended
  action $a_i$, because the recommendation $a_i$ will be revealed after the
  fact.  Any such deviation she could also have simulated under the
  original signaling mechanism, by first drawing $a_i$ from
  $\sigma_i(\cdot|m_i)$ and then changing to another action in the same SIS
  accordingly, which still would be an undetectable deviation under the original
  signaling mechanism.  But we know by assumption that that deviation was
  not beneficial under the original mechanism, because by assumption
  playing the distribution $\sigma_i(\cdot|m_i)$ was in equilibrium.
Hence, any deviation under the new signaling mechanism is also not beneficial.
\end{proof}

\end{document}